\documentclass[letterpaper, 10 pt, conference]{ieeeconf}

\IEEEoverridecommandlockouts  
\usepackage{amsmath,amssymb,amsfonts}

\usepackage{amsthm}
\usepackage{graphicx}
\usepackage{booktabs}
\usepackage{cite}
\usepackage{url}
\usepackage{algorithm}
\usepackage{algorithmic}

\graphicspath{{figures/}}

\newtheorem{lemma}{Lemma}
\newtheorem{proposition}{Proposition}

\theoremstyle{definition}
\newtheorem{definition}{Definition}
\newtheorem{problem}{Problem}

\title{\LARGE \bf
Certified Set Convergence for Piecewise Affine Systems
}

\author{Yanliang Huang$^{1,*}$, Peng Xie$^{1,*}$, Zhen Zhang$^{1}$, Wenyuan Wu$^{1}$, Zhuoqi Zeng$^{2}$, Amr Alanwar$^{1}$%
\thanks{$^{*}$Equal contribution.}%
\thanks{$^{1}$School of Computation, Information and Technology, Technical University of Munich, Munich, Germany.
{\tt\small \{yanliang.huang, p.xie, zhenzhang.zhang, wenyuan.wu, alanwar\}@tum.de}}%
\thanks{$^{2}$School of Engineering, Hainan Bielefeld University of Applied Sciences, Hainan, China.
{\tt\small zhuoqi.zeng@hibiuh.edu.cn}}%
}

\begin{document}

\maketitle
\thispagestyle{empty}
\pagestyle{empty}

\begin{abstract}
When all individuals of a system receive the same input, one open-loop input sequence has to steer every possible initial state to a target.
We study this problem for piecewise affine systems with bounded disturbances, where every state of an initial box must reach a robustly invariant target.
When the reachable set crosses a switching surface, it splits into fragments that follow different dynamics, and existing set-based controllers then give each fragment its own state feedback, which a common input does not allow.
Instead, we propose a planner-independent and sound verification framework for common-input control of disturbed piecewise affine systems.
Any planner can propose candidate sequences, and the verifier certifies a sequence by following every fragment under the dynamics of its own region.
To certify all possible initial boxes in a given range, we partition the range into finitely many cells and certify each cell.
On three evaluation systems, every possible initial box in the considered ranges is certified.
\end{abstract}

\section{Introduction}\label{sec:intro}

In some control problems, many individuals receive the same input.
For example, a large population of simple robots can be steered by one common signal \cite{becker2013massive}, and a global magnetic field exerts the same input on every particle of a swarm \cite{shahrokhi2019wall}.
In manipulation, tilting a tray likewise accelerates every part on it in the same way, which allows parts to be positioned without sensing \cite{erdmann1988sensorless}.
In all these cases the input cannot depend on the state of each individual, so the controller reduces to an open-loop input sequence.
This common sequence has to bring every possible initial state to the goal.

\begin{figure*}[t]
  \centering
  \includegraphics[width=0.788\textwidth]{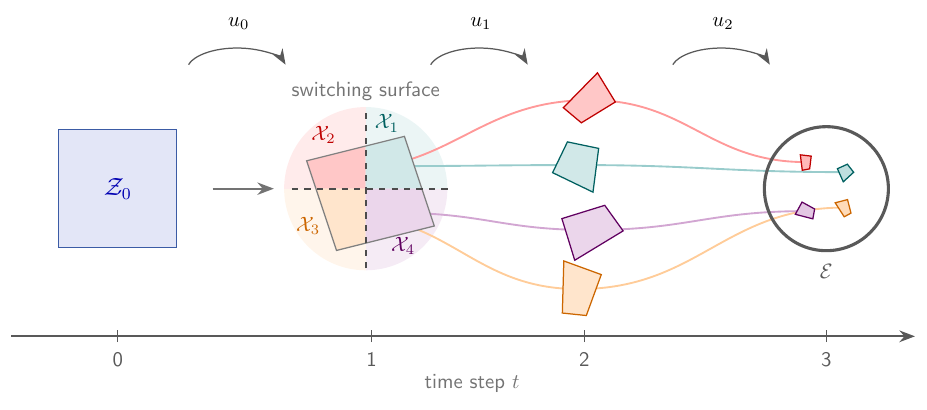}
  \caption{Set convergence under one common input. The initial box $\mathcal Z_0$ splits into fragments at a switching surface, and each fragment evolves under the affine dynamics of its own region $\mathcal X_i$, shown by color. The same input $u_t$ acts on every fragment at every step, and all fragments lie in the target $\mathcal E$ at step $k = 3$.}
  \label{fig:overview}
\end{figure*}

In such systems, contacts, flow zones and terrain often divide the state space into regions with different dynamics, so the same input moves states in different regions differently.
We therefore model these systems as piecewise affine (PWA) systems with bounded disturbances.
The task is to bring every state of an initial box into an ellipsoidal target within $K$ steps, without leaving a given domain, under every disturbance sequence.
The target is robustly invariant, so the states do not leave it afterwards.
The initial box itself is usually not known in advance, only the range of its center and a bound on its size, so we certify a whole collection of initial boxes.

The switching surfaces between regions make such a common sequence hard to plan.
When the set of states reachable from the initial box crosses a switching surface, it splits into fragments that follow different affine dynamics, as Fig.~\ref{fig:overview} illustrates.
The same input acts on all fragments at once and decides which region each of them enters next, so planning becomes combinatorial.
Existing set-based controllers for disturbed PWA systems use an open-loop sequence only while the reachable set stays in one region \cite{silva2003ecc, desimini2020robust}, and once the set splits, they give each fragment its own state feedback \cite{desimini2020robust, schurmann2020nahs}.
They therefore cannot certify one common sequence for a set that has to cross a switching surface on its way to the target.

For a set that splits, a common sequence is hard to plan with a guarantee, but checking whether a given sequence succeeds is tractable.
We therefore separate planning from certification, letting planners propose candidate sequences and leaving the guarantee to a verifier.
Given a candidate sequence, our verifier tracks every fragment of the reachable set under the dynamics of its own region and certifies the sequence when all fragments stay in the domain and lie in the target at the same step.
To certify a whole collection of initial boxes, we divide the range of box centers into finitely many cells and certify, for each cell, one enclosing box that contains all initial boxes centered in the cell.

\subsection{Related Work}

Robust control of PWA systems with bounded disturbances has been studied through dynamic programming \cite{kerrigan2002pwa}, optimization over region sequences \cite{silva2003ecc} and the propagation of zonotopic reachable sets \cite{desimini2020robust, schurmann2020nahs}.
Schmidt et al.\ \cite{schmidt2004cacsd} search for one input sequence for a set of initial states and propagate ellipsoids under the dynamics of every region they meet.
We instead split the reachable set at the switching surfaces, so that each fragment follows only the dynamics of its own region, and we certify a whole collection of initial boxes.

Reachable sets of PWA systems can be represented by constrained zonotopes, which are closed under affine maps, Minkowski sums and intersections with halfspaces \cite{scott2016constrained, raghuraman2022set}.
Hybrid zonotopes represent unions of constrained zonotopes and give reachable sets of hybrid systems \cite{bird2023hybrid, siefert2023robust}, and Xie et al.\ compute such sets for PWA systems from data \cite{xie2025datadriven}.
These works compute reachable sets under given inputs or control laws, whereas we use constrained zonotopes to certify a candidate input sequence.
Keeping the fragments as an explicit family lets our verifier check each fragment against $\mathcal D$ and $\mathcal E$ with convex programs, whereas checking a hybrid zonotope against $\mathcal E$ requires a mixed-integer program over its binary factors.

For linear systems, Fan et al.\ \cite{fan2018cav} and Sch\"urmann and Althoff \cite{schurmann2017acc} bring a whole set of initial states to a target by combining a reference input with a feedback that corrects the deviation of each state from the reference.
With a common input, each state cannot receive its own correction, so we certify open-loop sequences instead, and we do so for PWA systems.

Ensemble control also steers a whole family of systems with one open-loop input, for systems whose parameters differ across the family, such as the bilinear Bloch equations \cite{li2006pra, beauchard2010cmp}.
These works establish controllability and construct inputs for smooth dynamics without disturbances.
We instead certify a given input sequence for one PWA system with bounded disturbances.

\subsection{Contributions}

\begin{itemize}
  \item A planner-independent verifier that intersects the reachable set with every region it enters and propagates the fragments under one common input sequence, so that a sequence from any planner can be certified even when the reachable set splits.
  \item Certification by partitioning, which divides the range of box centers into cells and certifies the enclosing box of each cell, so that finitely many certificates cover a collection of infinitely many initial boxes.
  \item An evaluation on a synthetic eight-sector system and on two applications, a magnetically driven microrobot swarm and parts on a tilting tray.
  Every cell of the three collections is certified.
\end{itemize}

The remainder of this paper is organized as follows.
Section~\ref{sec:problem} introduces the set representations and states the problem, Section~\ref{sec:method} presents the verifier, the planners and certification by partitioning, Section~\ref{sec:experiments} evaluates the approach on the three systems, and Section~\ref{sec:conclusion} concludes the paper.

\section{Preliminaries and Problem Statement}\label{sec:problem}

\subsection{Notation and Set Representations}\label{sec:sets}

We start with the notation and the set representations that the verifier uses.
The identity matrix of dimension $d$ is denoted by $I_d$, and $\mathrm{diag}(a_1, \dots, a_d)$ is the diagonal matrix with diagonal entries $a_1, \dots, a_d$.
For a list $a$, the element $i$ of $a$ is denoted by $a^{(i)}$, and for a vector $z$, the component $j$ of $z$ is denoted by $z_j$.

\begin{definition}[Zonotope~\cite{kuhn1998rigorously}]\label{def:zono}
Given a dimension $d \in \mathbb{N}$, a center $c_\mathcal{Z} \in \mathbb{R}^{d}$ and $\gamma_\mathcal{Z} \in \mathbb{N}$ generator vectors in a generator matrix $G_\mathcal{Z} = \big[g_\mathcal{Z}^{(1)} \ \dots \ g_\mathcal{Z}^{(\gamma_\mathcal{Z})}\big] \in \mathbb{R}^{d \times \gamma_\mathcal{Z}}$, a zonotope is defined as
\begin{equation}\label{eq:zono}
\mathcal{Z} = \Big\{ x \in \mathbb{R}^{d} \ \Big|\ x = c_\mathcal{Z} + \sum_{i=1}^{\gamma_\mathcal{Z}} \beta_i g_\mathcal{Z}^{(i)},\ -1 \le \beta_i \le 1 \Big\} .
\end{equation}
We use the shorthand notation $\mathcal{Z} = \langle c_\mathcal{Z}, G_\mathcal{Z} \rangle$.
\end{definition}

The box with center $c \in \mathbb{R}^d$ and half-width vector $r \in \mathbb{R}^d_{\ge 0}$ is the zonotope $\mathcal{B}(c, r) = \langle c, \mathrm{diag}(r_1, \dots, r_d) \rangle$.
For a matrix $\Phi \in \mathbb{R}^{q \times d}$ and a vector $v \in \mathbb{R}^{q}$, the affine image of a zonotope is $\Phi\mathcal{Z} + v = \langle \Phi c_\mathcal{Z} + v, \Phi G_\mathcal{Z} \rangle$.
Given two zonotopes $\mathcal{Z}_1 = \langle c_{\mathcal{Z}_1}, G_{\mathcal{Z}_1} \rangle$ and $\mathcal{Z}_2 = \langle c_{\mathcal{Z}_2}, G_{\mathcal{Z}_2} \rangle$, we use $\oplus$ to denote the Minkowski sum, which is computed as
\begin{equation}\label{eq:minkowski}
\begin{aligned}
\mathcal{Z}_1 \oplus \mathcal{Z}_2 &= \{ z_1 + z_2 \mid z_1 \in \mathcal{Z}_1,\ z_2 \in \mathcal{Z}_2 \} \\
&= \big\langle c_{\mathcal{Z}_1} + c_{\mathcal{Z}_2}, [G_{\mathcal{Z}_1}, G_{\mathcal{Z}_2}] \big\rangle .
\end{aligned}
\end{equation}

Zonotopes are closed under affine maps and Minkowski sums, but the intersection of a zonotope with a halfspace is in general not a zonotope.
Linear equality constraints on the coefficients $\beta_i$ remove this limitation.

\begin{definition}[Constrained zonotope~\cite{scott2016constrained}]\label{def:czono}
Given a center $c_\mathcal{C} \in \mathbb{R}^{d}$, a generator matrix $G_\mathcal{C} \in \mathbb{R}^{d \times \gamma_\mathcal{C}}$, a constraint matrix $A_\mathcal{C} \in \mathbb{R}^{\kappa_\mathcal{C} \times \gamma_\mathcal{C}}$ and a constraint vector $b_\mathcal{C} \in \mathbb{R}^{\kappa_\mathcal{C}}$, a constrained zonotope is defined as
\begin{equation}\label{eq:czono}
\mathcal{C} = \Big\{ x \in \mathbb{R}^{d} \ \Big|\ x = c_\mathcal{C} + G_\mathcal{C} \beta,\ \|\beta\|_\infty \le 1,\ A_\mathcal{C} \beta = b_\mathcal{C} \Big\} .
\end{equation}
We use the shorthand notation $\mathcal{C} = \langle c_\mathcal{C}, G_\mathcal{C}, A_\mathcal{C}, b_\mathcal{C} \rangle$.
\end{definition}

Constrained zonotopes represent bounded polytopes and are closed under affine maps and Minkowski sums with zonotopes \cite{scott2016constrained} and under intersections with halfspaces \cite{raghuraman2022set}.
We call a constrained zonotope a fragment, and a family of fragments is a finite list of fragments that represents the union of the fragments in the list.

\subsection{Disturbed Piecewise Affine Systems}

We consider the discrete-time PWA system
\begin{equation}\label{eq:pwa}
  x_{t+1} = A_i x_t + B_i u_t + a_i + w_t, \qquad x_t \in \mathcal X_i,\ i \in \{1,\dots,M\},
\end{equation}
with state $x_t \in \mathbb R^n$, input $u_t \in \mathcal U = [-1,1]^m$ and disturbance
$w_t \in \mathcal W = \{w \in \mathbb R^n : \|w\|_\infty \le \varepsilon\}$.
The regions $\mathcal X_i = \{x : H_i x \le h_i\}$ are closed polyhedra that cover the bounded domain $\mathcal D$.

A state $x$ on a switching surface lies in every region of the index set $I(x) = \{i : x \in \mathcal X_i\}$, and we let its successor follow any of these dynamics.
The successors of a set $\mathcal Z \subseteq \mathcal D$ under the input $u$ are therefore
\begin{equation}\label{eq:successor}
  F(\mathcal Z, u) = \{A_i x + B_i u + a_i + w : x \in \mathcal Z,\ i \in I(x),\ w \in \mathcal W\}.
\end{equation}
For every choice of one region at each state on a switching surface, the resulting successors lie in $F(\mathcal Z,u)$, so every conclusion drawn from \eqref{eq:successor} holds for each such choice.

The target is the ellipsoid $\mathcal E = \{x : \|L(x - x_e)\|_2 \le \rho\} \subseteq \mathcal D$ with center $x_e \in \mathbb{R}^n$, radius $\rho > 0$ and an invertible matrix $L \in \mathbb{R}^{n \times n}$ that sets the shape of $\mathcal E$.
For $L = I_n$, the target $\mathcal E$ is the ball of radius $\rho$ around $x_e$.
We call $\mathcal E$ robustly invariant under the terminal input $\bar u \in \mathcal U$ if $F(\mathcal E, \bar u) \subseteq \mathcal E$.
For a state $x \in \mathcal E \cap \mathcal X_i$ and a disturbance $w \in \mathcal W$, adding and subtracting $A_i x_e$ in the successor $x^+ = A_i x + B_i \bar u + a_i + w$ and inserting $L^{-1}L$ give
\begin{equation*}
  \begin{aligned}
  L(x^+ - x_e) &= L A_i L^{-1}\, L(x - x_e) \\
  &\quad + L(A_i x_e + B_i \bar u + a_i - x_e) + L w .
  \end{aligned}
\end{equation*}
The triangle inequality and $\|L(x - x_e)\|_2 \le \rho$ then give
\begin{equation}\label{eq:invariance}
  \begin{aligned}
  \|L(x^+ - x_e)\|_2 &\le \|L A_i L^{-1}\|_2\,\rho \\
  &\quad + \|L(A_i x_e + B_i \bar u + a_i - x_e)\|_2 \\
  &\quad + \max_{w \in \mathcal W}\|L w\|_2 .
  \end{aligned}
\end{equation}
If the right side of \eqref{eq:invariance} is at most $\rho$ for every region that meets $\mathcal E$, then every successor of a state in $\mathcal E$ satisfies $\|L(x^+ - x_e)\|_2 \le \rho$, so $F(\mathcal E, \bar u) \subseteq \mathcal E$ and $\mathcal E$ is robustly invariant under $\bar u$.

\subsection{Problem Statement}

For an initial set $\mathcal Z_0 \subseteq \mathcal D$, an input sequence $u_0, \dots, u_{k-1}$ generates the reachable sets $\mathcal R_0 = \mathcal Z_0$ and $\mathcal R_{t+1} = F(\mathcal R_t, u_t)$, $t = 0, \dots, k-1$.

\begin{problem}\label{prob:steer}
Given an initial box $\mathcal Z_0 \subseteq \mathcal D$, a horizon $K$ and a target $\mathcal E$ that is robustly invariant under $\bar u$, find an entry step $k \le K$ and one input sequence $u_0, \dots, u_{k-1} \in \mathcal U$ such that $\mathcal R_t \subseteq \mathcal D$ for $t < k$ and $\mathcal R_k \subseteq \mathcal E$.
\end{problem}

Because $\mathcal E$ is robustly invariant under $\bar u$, a solution of Problem~\ref{prob:steer} steers every state of $\mathcal Z_0$ into $\mathcal E$ within $K$ steps and, with $\bar u$ applied from step $k$ on, keeps it there under every disturbance sequence.
We call the input sequence together with its entry step a certificate of $\mathcal Z_0$.

A set $\mathcal Q$ of centers and a set $\Gamma$ of half-width vectors define the collection of initial boxes
$\mathcal S(\mathcal Q, \Gamma) = \{\mathcal B(c, r) \subseteq \mathcal D : c \in \mathcal Q,\ r \in \Gamma\}$.
The collection is certified if every box in it has a certificate.

\section{Method}\label{sec:method}

\begin{figure}[t]
  \centering
  \includegraphics[width=\columnwidth]{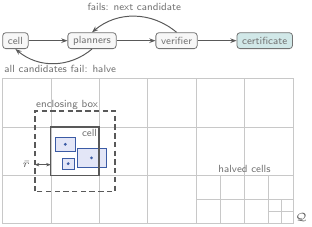}
  \caption{Certification by partitioning. The range $\mathcal Q$ of box centers is divided into cells. For each cell, the planners propose candidate sequences until the verifier certifies one for the enclosing box of the cell, and a cell for which every candidate fails is halved. The enclosing box extends the cell by $\bar r$ on every side, so its certificate covers every initial box centered in the cell, also one that reaches beyond the cell.}
  \label{fig:partition}
\end{figure}

Within one region the dynamics \eqref{eq:pwa} are affine, and the successors of two states $x, y \in \mathcal X_i$ under the same input and the same disturbance satisfy $x^+ - y^+ = A_i (x - y)$.
The shape of the successor set is therefore fixed by $A_i$, and the input only translates it.
If each reachable set lies in a single region and this sequence of regions is fixed, the generators of the reachable sets are independent of the inputs and their centers are affine in the inputs.
Problem~\ref{prob:steer} restricted to such a sequence is then a convex program \cite{desimini2020robust}, and for the ellipsoidal target $\mathcal E$ it is a second-order cone program.

A set that crosses a switching surface breaks this structure.
The set splits into fragments that evolve under different dynamics, so the same input moves the fragments by different amounts, as Fig.~\ref{fig:overview} illustrates.
The reachable set becomes a union whose number of pieces grows with every crossing, and the input sequences that keep this union in $\mathcal D$ and bring it into $\mathcal E$ in general form a nonconvex set.

Instead of solving this nonconvex problem directly, we certify a given input sequence with a verifier that propagates all fragments together and checks the result against $\mathcal D$ and $\mathcal E$.
The certificate depends only on this propagation, so it applies to input sequences from any planner.

\subsection{The Verifier}\label{sec:verifier}

A certificate requires that the reachable sets stay in $\mathcal D$ up to the entry step and that the reachable set at the entry step lies in $\mathcal E$.
Since the reachable set splits at switching surfaces, the verifier represents $\mathcal R_t$ as the union of a family $\mathcal F_t$ of fragments, starting from $\mathcal F_0 = \{\mathcal Z_0\}$.
Each step of the verifier propagates the family and then checks it against $\mathcal D$ and $\mathcal E$, as summarized in Algorithm~\ref{alg:verifier}.

The propagation keeps this representation exact.
The successor set \eqref{eq:successor} moves every state with the dynamics of each region that contains it, so for every fragment $\mathcal C$ it splits along the regions,
\begin{equation}\label{eq:split}
  F(\mathcal C, u_t) = \bigcup_i \mathcal C_i^+ , \qquad \mathcal C_i^+ = A_i(\mathcal C \cap \mathcal X_i) + B_i u_t + a_i \oplus \mathcal W .
\end{equation}
Each set $\mathcal C_i^+$ is again a constrained zonotope.
The verifier therefore forms $\mathcal F_{t+1} = \{\mathcal C_i^+ : \mathcal C \in \mathcal F_t,\ \mathcal C \cap \mathcal X_i \neq \emptyset\}$, which keeps the sets $\mathcal C_i^+$ separate.
By \eqref{eq:split}, the union of $\mathcal F_{t+1}$ equals $\bigcup_{\mathcal C \in \mathcal F_t} F(\mathcal C, u_t)$, which is the successor set of the union of $\mathcal F_t$.
Since $\mathcal F_0 = \{\mathcal Z_0\}$ has the union $\mathcal R_0$, induction over $t$ shows that the union of $\mathcal F_t$ is $\mathcal R_t$ for every $t$.

The verifier then checks the two certificate conditions on every fragment.
The fragment lies in $\mathcal D$ if it lies in the halfspace of every facet of $\mathcal D$, which takes one linear program per facet \cite{scott2016constrained}.
The fragment lies in $\mathcal E$ if its largest distance $\|L(x - x_e)\|_2$ to the center of $\mathcal E$ is at most $\rho$, and this largest distance is hard to compute because it is the maximum of a convex function over a polytope.
The following lemma bounds its square from above through an S-procedure \cite{boyd1994linear}, so the fragment lies in $\mathcal E$ whenever this bound is at most $\rho^2$.

\begin{lemma}\label{lem:sdp}
Let $\mathcal C = \langle c_\mathcal C, G_\mathcal C, A_\mathcal C, b_\mathcal C \rangle$, $v = L(c_\mathcal C - x_e)$ and $V = L G_\mathcal C$, and let $\mathbf 1$ denote the vector of ones.
If a scalar $\tau$, a vector $\lambda \ge 0$ with one entry per generator and a vector $\mu$ with one entry per equality constraint satisfy
\begin{equation}\label{eq:lmi}
  \Lambda(\tau, \lambda, \mu) =
  \begin{bmatrix}
    \tau - \mathbf 1^\top \lambda - 2 b_\mathcal C^\top \mu & \mu^\top A_\mathcal C & v^\top \\
    A_\mathcal C^\top \mu & \mathrm{diag}(\lambda) & V^\top \\
    v & V & I
  \end{bmatrix} \succeq 0 ,
\end{equation}
then $\|L(x - x_e)\|_2^2 \le \tau$ for every $x \in \mathcal C$.
\end{lemma}

\begin{proof}
By a Schur complement, \eqref{eq:lmi} states that the block formed by the first two block rows and columns of $\Lambda$ dominates $[v\ V]^\top [v\ V]$.
Evaluating both quadratic forms at $(1, \beta)$ gives $\|v + V\beta\|_2^2 \le \tau - \sum_i \lambda_i (1 - \beta_i^2) + 2\mu^\top(A_\mathcal C\beta - b_\mathcal C)$.
Every $x \in \mathcal C$ equals $c_\mathcal C + G_\mathcal C \beta$ for some $\beta$ with $\|\beta\|_\infty \le 1$ and $A_\mathcal C\beta = b_\mathcal C$, and for such $\beta$ the last two terms are at most zero.
\end{proof}

The smallest bound that Lemma~\ref{lem:sdp} provides is the value of the semidefinite program
\begin{equation}\label{eq:sdp}
  \tau^\star(\mathcal C) = \min \{\tau : \Lambda(\tau, \lambda, \mu) \succeq 0,\ \lambda \ge 0\} .
\end{equation}

\begin{algorithm}[t]
\caption{Verifier}\label{alg:verifier}
\begin{algorithmic}[1]
\REQUIRE initial box $\mathcal Z_0$, inputs $u_0, \dots, u_{K-1}$
\ENSURE entry step $k$, or no certificate
\STATE $\mathcal F_0 \gets \{\mathcal Z_0\}$
\FOR{$t = 0, \dots, K-1$}
  \STATE $\mathcal F_{t+1} \gets \{\mathcal C_i^+ : \mathcal C \in \mathcal F_t,\ \mathcal C \cap \mathcal X_i \neq \emptyset\}$ \COMMENT{by \eqref{eq:split}}
  \IF[by \eqref{eq:sdp}]{$\tau^\star(\mathcal C) \le \rho^2$ for all $\mathcal C \in \mathcal F_{t+1}$}
    \RETURN $k = t + 1$
  \ENDIF
  \IF[by LPs]{$\mathcal C \not\subseteq \mathcal D$ for some $\mathcal C \in \mathcal F_{t+1}$}
    \RETURN no certificate
  \ENDIF
\ENDFOR
\RETURN no certificate
\end{algorithmic}
\end{algorithm}

\begin{proposition}[Soundness]\label{prop:sound}
Suppose that Algorithm~\ref{alg:verifier}, applied to an initial box $\mathcal Z_0$ and inputs $u_0, \dots, u_{K-1}$, returns an entry step $k$.
Then the inputs $u_0, \dots, u_{k-1}$ and the entry step $k$ form a certificate of $\mathcal Z_0$.
\end{proposition}

\begin{proof}
The union of $\mathcal F_t$ is $\mathcal R_t$ for every $t$, and $\mathcal R_0 = \mathcal Z_0$ lies in $\mathcal D$ by assumption.
For $0 < t < k$, Algorithm~\ref{alg:verifier} has found every fragment of $\mathcal F_t$ in $\mathcal D$, and at $t = k$ Lemma~\ref{lem:sdp} places every fragment of $\mathcal F_k$ in $\mathcal E$.
\end{proof}

By Proposition~\ref{prop:sound}, every entry step that the verifier returns yields a solution of Problem~\ref{prob:steer}.

\subsection{Candidate Sequences}\label{sec:candidates}

By Proposition~\ref{prop:sound}, an initial box has a certificate as soon as one candidate sequence passes the verifier, whichever planner proposed it.
We therefore collect candidates from three kinds of planners, namely single-mode tree search, one-step contraction and long-horizon planning.
Each candidate that passes the verifier certifies one box, and we call the certificates from all candidates together the union certificate, which covers every box for which at least one candidate passes.
The three kinds differ in how they handle the nonconvexity of planning, which arises because the future splits of the reachable set depend on the inputs.

If every reachable set lies in a single region, the problem for a fixed sequence of regions is convex.
Single-mode tree search uses this structure and plans all inputs at once, searching over prefixes of region sequences with a linear feasibility program at each node and a second-order cone program at full depth.
When the reachable sets split, one step already involves several regions, so a sequence with one region per step no longer describes the reachable sets.
Over a single step, however, the split of the current family is already known, and every successor moves affinely with the input while keeping its generators.
One-step contraction uses this structure and chooses one input at a time.
Every successor is enclosed in a ball whose radius does not depend on the input.
A second-order cone program then minimizes the largest distance from a point of these balls to an aim point, which is the center $x_e$ of $\mathcal E$ or a point close to it.
This choice is greedy with a horizon of one step, so it can stop at a local minimum of the distance to the aim point, where no single input brings the family closer.
Leaving such a minimum requires a longer horizon, which long-horizon planners provide either by scoring whole input sequences or through a policy trained over the whole horizon.
We use sampling planners and value-guided planning.
Sampling planners draw input sequences at random, score them by simulating the piecewise dynamics or by running the verifier, and refine the best ones.
Value-guided planning uses a policy that is trained by reinforcement learning on the piecewise dynamics with a reward for reaching the target.
The planner applies this policy at the center of the family or uses it to set the aim point of one-step contraction.
The guarantee comes from the verifier alone, so the long-horizon planners contribute certificates of the same strength as the planners built on convex programs.

\begin{table*}[t]
  \centering
  \caption{Parameters of the three systems. All systems use the terminal input $\bar u = 0$, and every cell is halved at most twice.}
  \label{tab:systems}
  \begin{tabular}{llll}
    \hline
     & Eight-sector & Swarm & Tray \\
    \hline
    $n$, $m$, $M$ & 2, 2, 8 & 2, 2, 5 & 4, 2, 5 \\
    $\varepsilon$ & 0.001 & 0.01 & 0.005 \\
    $\mathcal D$ & $[-3,3]^2$ & $[-4,4]\times[-3,3]$ & $[-4,4.4]\times[-3,3]\times[-1.5,1.5]^2$ \\
    $x_e$ & $1.2(\cos\frac{\pi}{8},\sin\frac{\pi}{8})$ & $(3.2,0)$ & $(3.2,0,0,0)$ \\
    $\rho$ & 0.391 & 0.5 & 0.5 \\
    right side of \eqref{eq:invariance} & 0.3103 & 0.3141 & 0.4432 \\
    $K$ & 15 & 30 & 65 \\
    $\mathcal Q$ & $[-2.7,2.7]\times[-2.7,2.5]$ & $[-3.5,-0.5]\times[-2.5,2.5]$ & $[-3.5,-0.5]\times[-2.5,2.5]\times\{0\}^2$ \\
    $\bar r$ & $(0.15,0.15)$ & $(0.3,0.3)$ & $(0.3,0.3,0.02,0.02)$ \\
    first cell half-width & 0.15 & 0.25 & 0.25 \\
    \hline
  \end{tabular}
\end{table*}

\subsection{Certification by Partitioning}\label{sec:partition}

The collection contains infinitely many boxes, and to enclose all boxes whose centers lie in one cell, we bound their size by fixing the largest half-widths $\bar r$ and certify the collection $\mathcal S(\mathcal Q, \Gamma)$ with $\Gamma = \{r : 0 \le r \le \bar r\}$.
We partition the set $\mathcal Q$ of centers into cells, certify one box per cell and halve every side of a cell that fails, until all cells are certified or a fixed depth is reached.
We call this procedure certification by partitioning.
Fig.~\ref{fig:partition} illustrates the procedure.
Because the successor set \eqref{eq:successor} is the union of the successors of the individual states, the reachable sets of an inner box lie in those of an outer box that contains the inner box under the same inputs.
A certificate of the outer box is therefore also a certificate of the inner box.
Since $\Gamma$ bounds the half-widths by $\bar r$, every box of the collection whose center lies in a cell $\mathcal B(q, h)$ lies in the box $\mathcal B(q, h + \bar r)$, which we call the enclosing box of the cell.
If Algorithm~\ref{alg:verifier} returns an entry step for the enclosing box, its inputs therefore certify every box of the cell, also when the enclosing box itself extends beyond $\mathcal D$.
The verifier therefore runs once per cell and candidate, with candidates from the three kinds of planners.

\begin{proposition}[Certification of a collection]\label{prop:region}
Let the cells $\mathcal B(q_j, h_j)$, $j = 1, \dots, J$, cover $\mathcal Q$, and let Algorithm~\ref{alg:verifier} return an entry step for every enclosing box $\mathcal B(q_j, h_j + \bar r)$.
Then the collection $\mathcal S(\mathcal Q, \Gamma)$ is certified, and every box in the collection has the certificate formed by the inputs and the entry step returned for the enclosing box of a cell that contains the center of the box.
\end{proposition}

\begin{proof}
A box of the collection lies in $\mathcal D$ and in the enclosing box of a cell that contains its center.
Under the returned inputs, the reachable sets of the box lie in the unions of the fragment families of the enclosing box, which Algorithm~\ref{alg:verifier} has found in $\mathcal D$ for $0 < t < k$ and in $\mathcal E$ at $t = k$.
Together with the box itself lying in $\mathcal D$, the box satisfies both conditions of the certificate.
\end{proof}

\section{Experiments}\label{sec:experiments}

\begin{figure*}[t]
  \centering
  \includegraphics[width=\textwidth]{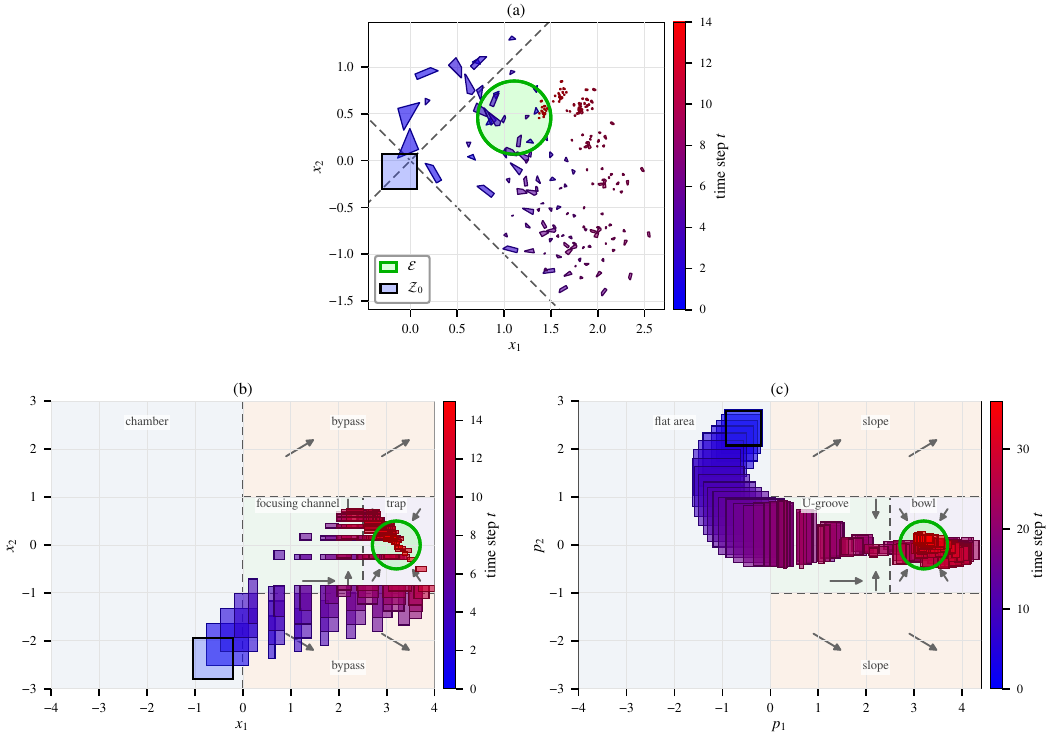}
  \caption{Certificates on the three systems, (a) the eight-sector system, (b) the swarm and (c) the tray. In each panel, the initial box $\mathcal Z_0$ is the enclosing box of one cell, the polygons are the fragments of the reachable sets, colored by the time step, and the dashed lines are the switching surfaces. In (b) and (c), the shading marks the regions and the arrows show the direction in which each region pushes the state, and (c) shows the projection onto the position plane. The cells are chosen to show splitting, with an enclosing box that contains the origin in (a) and initial boxes that start next to a bypass in (b) and next to a slope in (c).}
  \label{fig:certs}
\end{figure*}

\subsection{Systems and Setup}\label{sec:setup}
We evaluate certification by partitioning on a synthetic eight-sector system and two applications, with the parameters in Table~\ref{tab:systems}.
All three systems use $\bar u = 0$, and the right side of \eqref{eq:invariance} is smaller than $\rho$ in all three systems, so $\mathcal E$ is robustly invariant.

The eight-sector system has eight regions, the sectors bounded by the coordinate axes and the two diagonals, which meet at the origin.
Fig.~\ref{fig:certs}(a) shows these sectors.
In each region the dynamics rotate the state about an equilibrium of the region and contract the distance to the equilibrium by the factor 0.79, and the input shifts the state by $0.5u$.
The target $\mathcal E$ is a disk around the equilibrium of the first region.
The centers of the collection cover nearly all of $\mathcal D$, and the first propagation step splits every box that contains the origin in its interior into eight fragments.

The first application steers a swarm of magnetic microrobots through a microfluidic chip.
A global magnetic field gradient moves all robots of the swarm together, so the robots receive the same input, and the swarm is the initial box of Problem~\ref{prob:steer}.
At the low Reynolds number of the chip the robots have no inertia, so the state is the position in millimeters.
One step lasts 1\,s, the input displaces every robot by $0.3u$, so $B_i = 0.3I$ in all regions, and the disturbance represents Brownian motion.
The five regions are the flow zones of the chip.
In the chamber $x_1 \le 0$ the fluid is at rest, so the chamber dynamics are $x^+ = x + 0.3u + w$.
The two bypasses $x_1 \ge 0$, $|x_2| \ge 1$ carry the robots 0.25 downstream and 0.15 sideways per step.
The focusing channel $0 \le x_1 \le 2.5$, $|x_2| \le 1$ carries the robots 0.25 downstream and contracts $x_2$ by the factor 0.6.
The trap $x_1 \ge 2.5$, $|x_2| \le 1$ contracts the state toward $x_c = (3.2, 0)$ by the factor 0.6.
Fig.~\ref{fig:certs}(b) shows the five zones and the direction in which each zone pushes the robots.
The target $\mathcal E$ is a disk around $x_c$.
The centers of the collection cover the chamber, so every swarm starts in the chamber and has to cross the surface $x_1 = 0$ to reach the trap.
In the chamber one step moves a robot by at most 0.31 in $x_1$.
A swarm wider than 0.31 therefore lies on both sides of the surface for at least one step, and the reachable set splits at the surface.

The second application moves parts on a tilting tray into a bowl.
The tilt accelerates all parts equally, so the parts receive the same input.
The state of a part consists of the position $p$ and the velocity $v$, and the tilt $u$ adds the acceleration $\beta u$ with $\beta = 0.15$.
On the flat area the dynamics are $p^+ = p + v$ and $v^+ = a v + \beta u$, where the factor $a = 0.5$ models linear damping, for example by an oil film under the parts.
The regions depend on the position alone and have the layout of the chip.
The flat area $p_1 \le 0$ corresponds to the chamber, and the two slopes add a constant acceleration.
The U-groove drives the speed along the groove toward 0.5 and pulls the parts back to the groove axis, and the bowl pulls the parts toward its center $(3.2, 0)$.
Fig.~\ref{fig:certs}(c) shows this layout.
The target $\mathcal E$ is an ellipsoid around the resting state $x_e = (3.2, 0, 0, 0)$.
The matrix $L$ comes from the real Jordan basis of the bowl dynamics, in which the bowl dynamics contract, and the ellipsoid has the radius 0.5 in position and 0.2236 in velocity.
The parts start almost at rest, with velocities in the box of half-width 0.02 around zero, and their positions form the same collection as the swarm.
Since $a = 0.5$ and $\beta = 0.15$, the speed on the flat area stays below 0.31, so one step moves a part by less than 0.32 in $p_1$.
A group of parts wider than 0.32 therefore also splits at the surface $p_1 = 0$.

The verifier represents the fragments as constrained zonotopes in CORA~\cite{althoff2015cora}.
To bound the cost of one step, the implementation merges fragments when a family exceeds 300 fragments and reduces the generators and constraints of large fragments.
Both operations replace fragments by sets that contain the fragments, so the union of the family still contains $\mathcal R_t$, and Proposition~\ref{prop:sound} still holds.
\begin{table*}[t]
  \centering
  \caption{Planners that propose candidate sequences.}
  \label{tab:planners}
  \begin{tabular}{p{5.7cm}p{7.5cm}p{3.3cm}}
    \hline
    Planner & Main parameters & Systems \\
    \hline
    Single-mode tree search & Node budget 600 & All three \\
    One-step contraction & Aim point $x_e$, on the eight-sector system also $(1.2, 0.4)$ & All three \\
    Sampling, random-disturbance scoring & 60 rounds of 300 samples, keeping the best 30 & Eight-sector, tray \\
    Sampling, worst-case scoring & As above, adding the 16 sequences that hold $w$ at one vertex of $\mathcal W$ & Tray \\
    Sampling, verifier scoring & Up to 6 rounds of 16 samples, keeping the best 4 & Eight-sector, tray \\
    Value-guided planning & Policy trained with the deep deterministic policy gradient algorithm \cite{lillicrap2016ddpg} & Eight-sector \\
    \hline
  \end{tabular}
\end{table*}

Table~\ref{tab:planners} lists the planners, their main parameters and the systems on which they run.
The linear programs and second-order cone programs are solved with linprog and coneprog of MATLAB R2026b, and the semidefinite program \eqref{eq:sdp} with SeDuMi~\cite{sturm1999sedumi} through YALMIP~\cite{lofberg2004yalmip}, on an AMD Ryzen 9 9950X processor.

\begin{figure*}[t]
  \centering
  \includegraphics[width=\textwidth]{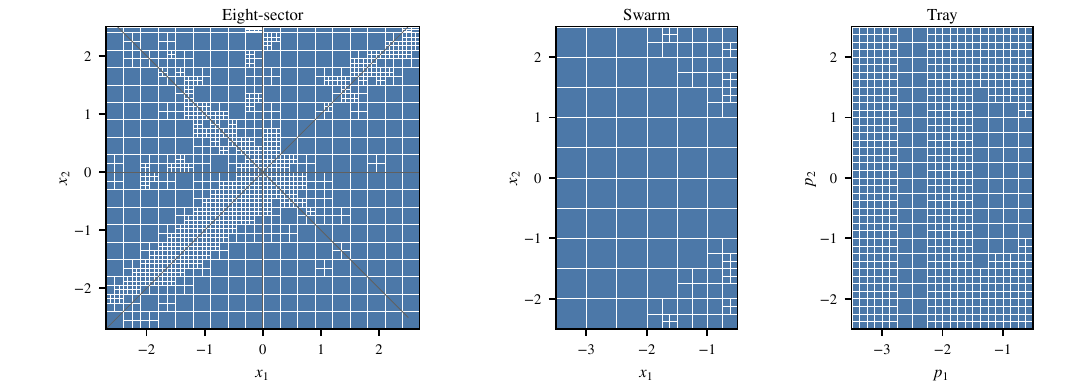}
  \caption{Final cells of the three collections, all of them certified. A smaller cell was halved more often.}
  \label{fig:cells}
\end{figure*}

\subsection{Results}\label{sec:results}

Every cell of the three collections is certified, so by Proposition~\ref{prop:region} every box of the three collections has a certificate.
The collections have 1296, 120 and 720 cells on the eight-sector system, the swarm and the tray.
Fig.~\ref{fig:certs} shows what such a certificate looks like.
In all three certificates the fragments spread apart and take different routes before they gather in $\mathcal E$.
On the eight-sector system the eight fragments created at the origin are pushed apart and gather again in $\mathcal E$.
On the swarm the fragments reach the trap both through the bypass and through the focusing channel.
On the tray the inputs first move the parts away from the slope, and the reachable set then enters the groove.
In all three cases one input sequence brings every fragment into $\mathcal E$ at the same step, which is 14, 15 and 36.

Fig.~\ref{fig:cells} shows which cells were halved.
On the eight-sector system these cells lie near the diagonals and the origin.
There the enclosing box meets several sectors, so the reachable set splits at the first step into fragments that rotate about different equilibria and drift apart, whereas every cell whose enclosing box touches only one sector is certified by tree search without halving.
On the swarm the hard cells lie in the corners of the chamber that face the bypasses.
A swarm that starts there reaches the exit at the height of a bypass, so its reachable set splits early, and the bypass carries one part sideways while the focusing channel draws the other part in.
On the tray every cell of the initial partition, with half-width 0.25, is halved, and among the smaller cells the hard ones lie at the height of the slopes, which push the parts away from the groove, and in most columns far from the groove.
Halving shrinks the enclosing box and with it every reachable set.
Certifying the collections took 6335, 260 and 3506 verifier runs on the eight-sector system, the swarm and the tray.
A verifier run that returns a certificate takes a median of 3.2, 9.0 and 53.6\,s.

On every cell the planners run in a fixed order until one candidate passes, and the cell keeps this first certificate.
On the swarm, one-step contraction certifies all 120 cells.
On the tray, one-step contraction certifies 354 of the 720 cells, and sampling with random-disturbance, worst-case and verifier scoring certifies the remaining 146, 203 and 17 cells.
On the eight-sector system, tree search, one-step contraction and value-guided planning certify 300, 419 and 163 cells, and sampling with random-disturbance and verifier scoring certifies the remaining 404 and 10.
No single planner certifies the whole collection of the eight-sector system or of the tray, and because the guarantee comes from the verifier alone, the certificates of different planners can be combined.

Intersecting each fragment exactly with the regions keeps the reachable sets tight.
Under the same input sequences, a variant of the verifier that replaces each intersection by its bounding box lets the family reach $\mathcal E$ at a later step for 971 cells.
For 10 cells, the variant does not certify the sequence at all.

The open-loop sequences of the set-based controllers in \cite{silva2003ecc, desimini2020robust} require the reachable set to stay in a single region at every step.
On the swarm and the tray no cell of the collections meets this requirement.
The cells are halved at most twice from the half-width 0.25, so the smallest cells have half-width 0.0625, and the enclosing box of a cell extends the cell by $\bar r = 0.3$ on every side.
Every enclosing box is therefore at least 0.725 wide.
A box wider than 0.31 on the swarm or 0.32 on the tray splits under every input sequence when it leaves the chamber or the flat area.
No sequence therefore keeps the reachable set of an enclosing box in one region, and the union certificate reaches these cells only through sequences under which the reachable set splits.

A sampling planner with random-disturbance scoring checks a candidate by simulating a grid of initial states under one random disturbance sequence.
Problem~\ref{prob:steer}, however, requires all states of the initial box to stay in $\mathcal D$ and reach $\mathcal E$ under every disturbance sequence, and the tray shows that a candidate can pass this check and still fail.
We draw 40 tray cells at random among those for which the planner found no certified sequence and simulate the planner's sequences from initial states in the enclosing boxes.
Under random disturbances the states stay in $\mathcal D$ on all 40 cells, whereas under deterministic disturbances, such as one held at a vertex of $\mathcal W$ in every step, the states leave $\mathcal D$ on all 40 cells.
Nothing on the flat area or in the groove pulls the position $p_1$ back, so a constant disturbance, as from a tray slightly out of level, adds up over the steps, whereas random disturbances largely cancel.

\section{Conclusion}\label{sec:conclusion}
We steer every state of an initial box of a disturbed PWA system into a robustly invariant ellipsoidal target with one common open-loop input sequence.
We certify such a sequence with a planner-independent verifier that intersects the reachable set with every region it enters and propagates the whole family of fragments.
To certify a collection of initial boxes at once, we partition the range of their centers into cells and certify the enclosing box of each cell.
With candidate sequences from three kinds of planners, we evaluated this certification by partitioning on a synthetic eight-sector system and two applications.
Every cell of the three collections is certified.

It remains open how to decide quantitatively, before any planning, whether a given range of initial boxes can be certified, for instance from the contraction within the regions and the geometry of the switching surfaces.

\bibliographystyle{IEEEtran}
\bibliography{references}

\end{document}